\documentclass[a4paper, 10pt]{article}
\usepackage{graphicx}
\usepackage{alltt}
\usepackage{amsmath}
\usepackage{amssymb}
\usepackage{amsthm}
\usepackage{authblk}
\usepackage{xcolor}
\usepackage{algorithm}
\usepackage{algpseudocode}
\usepackage{multirow}
\usepackage{fancyhdr}

\usepackage[hidelinks, pdftex]{hyperref}

\newtheorem{theorem}{Theorem}[section]  
\newtheorem{lemma}[theorem]{Lemma}      
\newtheorem{prop}[theorem]{Proposition}

\newtheorem{prob}[theorem]{Problem}
\newtheorem{deff}[theorem]{Definition}

\title{\LARGE
\textbf{Probabilistic verification algorithm for linear codes}}
\author[1]{\small\textbf {{Mingchao Li*}}}
\author[1]{\small\textbf {{Jiyou Li{$^\dagger$}}}}
\affil[1]{\small School of Mathematical Sciences, Shanghai Jiao Tong University, Shanghai, P.R.China}
\date{}

\begin{document}

\maketitle
\thispagestyle{fancy} 
\fancyhf{} 
\fancyhead{} 
\fancyfoot{} 
\lfoot{*: lmc3051@sjtu.edu.cn, $\dagger$: lijiyou@sjtu.edu.cn}
\renewcommand{\headrulewidth}{0pt}
\renewcommand{\footrulewidth}{0.4pt}

\begin{abstract}
In this paper, we propose a probabilistic algorithm suitable for any linear code $C$ to determine whether a given vector $\mathbf{x}$ belongs to $ C$. The algorithm achieves $O(n\log n)$ time complexity, $ O(n^2)$ space complexity. The error probability is less than $1/\mathrm{poly}(n)$ in the asymptotic sense.

 \vskip 2mm


\end{abstract}

\section{Introduction}

Linear codes have been widely used over the past decades. Significant progress has been made in developing decoding algorithms for specific code families, including the Berlekamp-Massey algorithm \cite{1054260} for Reed–Solomon (RS), BCH, and Goppa codes; modular approach and Fast Fourier Transform (FFT) for RS codes \cite{9925209}; and majority voting method \cite{179340} for algebraic geometry (AG) codes.   Before initiating the decoding procedure, one must check whether the received codeword contains errors. If no errors are detected, decoding is unnecessary. If errors are found, mechanisms such as retransmission systems can be used to minimize the number of times the decoder is invoked. Consequently, fast error detection is essential in practical applications.  Furthermore, certain codes that admit very efficient membership tests---known as locally testable codes (LTCs)---are also intrinsically connected to the Probabilistically Checkable Proofs (PCP).

Previous work has primarily focused on RS codes verification and the study of LTCs. Since RS codes can be viewed as evaluations of a polynomial on a proper subset, verifying them essentially amounts to recovering the coefficients from received vector. The primary technique used for this purpose is the FFT technique over finite fields. The reader is referred to  \cite{4395269} \cite{6176255} \cite{5910103} \cite{9925209} \cite{doi:10.1137/1.9781611977912.135} \cite{5625613} \cite{pollard1971fast} for more details about multiplicative FFT, additive FFT, cyclotomic FFT and Galois based FFT (G-FFT). LTCs can also be checked quickly. They are error-correcting codes equipped with probabilistic testers that can verify membership using only a small number of queries to the received vector \cite{10.1145/1162349.1162351}. The tester should accept codewords with probability one, and reject words that are far
(in Hamming distance) from the code with noticeable probability \cite{10.1007/978-3-642-15369-3_50}. Some well known codes, such as the Hadamard code, binary Reed-Muller (RM) codes \cite{1522661}, almost-orthogonal codes (including Dual-BCH, Goppa and  trace subcodes of AG codes) \cite{1530724}, are LTCs. However, many other codes are not LTCs, such as random codes and certain types of AG codes \cite{chen2009testing} \cite{10.1007/978-3-642-15369-3_50}.

While much research focuses on fast verification algorithms for specific linear codes, the problem of accelerating codeword verification for general linear codes has not been fully investigated. If the linear code lacks a rich algebraic structure, codeword verification may sometimes become the computational bottleneck. A typical example is the list decoding algorithm for RM codes. In 2004, Pellikaan and Wu \cite{1278668} viewed RM codes as subfield subcodes of RS codes and proposed a two-step list decoding method: the first step calls an RS code list decoder, and the second step checks each codeword in the resulting list to determine whether it belongs to the RM code. Although an algorithm proposed by Cohn and Heninger \cite{Cohn2015} in 2015 achieved quasi-linear time complexity for the first step under certain assumptions, the second step still required computing the syndrome, which takes $O(n^2)$ time. Recently, Kuang et al.\cite{11197050} optimized the second step to quasi-linear time by introducing a new basis for RM codes. The key to their algorithm lies in the fact that the transformation between the new basis and the standard basis is quasi-linear.

 We give some notation before stating the main problem. Let $C$ be any $[n,k]$ linear code over finite field $\mathbb{F}_q$ and let $H$ be its parity check matrix. Assuming a vector $\mathbf{x}\in \mathbb{F}_q^n$ is received, we need to determine whether $\mathbf{x}\in C$ as quickly as possible. Clearly, we cannot expect to directly compute $H\mathbf{x}^T$ in quasi-linear time, as that would imply an improvement in the complexity of matrix multiplication. Instead, we aim to design a probabilistic algorithm just like the idea behind LTCs.. Let us briefly analyze the desired parameters of our algorithm. Here, all operations are multiplications or additions over the finite field $\mathbb{F}_q$, and storage space is measured by the number of elements in $\mathbb{F}_q$. Our primary goal for the algorithm is to produce a result within $ O(n\log n)$ time, since directly computing $H\mathbf{x}^T$ requires $O(n^2)$ operations. Secondly, as the parity check matrix requires storing $n(n-k$) (or $k(n-k)$ in the systematic case) elements, the algorithm's storage space must be bounded by $O(n^2)$. In other words, compared to the standard method, we can tolerate at most a constant-factor overhead in storage.  Finally, we aim for an algorithm whose error probability is bounded by $1/\mathrm{poly}(n)$. If the error probability were too high, the algorithm would  be of limited practical utility. We thus arrive at the following central problem:
\begin{prob}
\label{prob1.1}
    Is it possible to design a probabilistic algorithm that decides whether $\mathbf{x}\in C$ in $O(n \log n)$ time, using $O(n^2)$ space, and with an error probability bounded by $1/\mathrm{poly}(n)$?
\end{prob}

Let $\mathbf{h}_1, \dots, \mathbf{h}_{n-k}$ denote the row vectors of the parity check matrix $H$. 
To prove $\mathbf{x} \notin C$, it suffices to find a row vector $\mathbf{h}_i$ that is not orthogonal to $\mathbf{x}$. 
A natural approach is to randomly sample one of the row vectors $\mathbf{h}_1, \dots, \mathbf{h}_{n-k}$, 
compute its inner product with $\mathbf{x}$, and aim to quickly identify such a non-orthogonal vector. Unfortunately, this straightforward strategy generally not works very well. 
Since $n>n -k$ and $H$ has full column rank, there exists a nonzero vector $\mathbf{e}$ such that the weight of $H\mathbf{e}^T$ equals 1. Suppose $\mathbf{x} = \mathbf{c} + \mathbf{e}$ for some $\mathbf{c} \in C$. 
Then the probability of selecting a row vector not orthogonal to $\mathbf{x}$ is only $1/(n-k)$. On average, we must independently perform $n-k$ inner product operations to detect $\mathbf{e}$, which provides no improvement compared to directly computing syndrome. The situation becomes even worse in the systematic case. 
Assume $H$ is given in the block form $[H_1\;|\; I_{n-k}]$. Even to detect a single burst error occurring in the last $n-k$ positions 
(corresponding to the identity submatrix), one would need to compute all inner products $(\mathbf{h}_1, \mathbf{x}), \dots, (\mathbf{h}_{n-k}, \mathbf{x})$.

Although this naive randomized method fails to achieve the desired efficiency, 
it highlights the importance of amplifying the proportion of vectors that are not orthogonal to $\mathbf{x}$. 
This observation motivates the algorithm to be presented in the next section.

\section{Algorithm design}
A natural generalization of the collection of row vectors in the parity check matrix is the following concept of a test set.
\begin{deff}
    Let $C$ be a $[n,k]$ linear code in $\mathbb{F}_q^n$ and $\mathbb{F}_{q^u}$ be the extension field of $\mathbb{F}_q$ of degree $u$. A subset $S$ of $\mathbb{F}_{q^u}^n$ is a \textbf{test set} with \textbf{designed probability} $p\in (0,1)$ for the code $C$ if for every $\mathbf{x} \in \mathbb{F}_q^n$:
\begin{enumerate}
    \item $\mathbf{x} \in C$ if and only if $(\mathbf{x}, \mathbf{y}) = 0$ for all $\mathbf{y} \in S$.
    \item If there exists $\mathbf{y} \in S$ such that $(\mathbf{x}, \mathbf{y}) \neq 0$, then 
    \begin{equation}
        \sharp\{\mathbf{y}\in S\;|\; (\mathbf{x},\mathbf{y})\not = 0\}\ge (1-p)\sharp S.\label{1}
    \end{equation}
\end{enumerate}
\end{deff}
Note that the second condition is equivalent to $\sharp\{\mathbf{y}\in S\;|\; (\mathbf{x},\mathbf{y}) = 0\}\le p\sharp S$. The test set $S$ and the original code $C$ have the following relationship.
\begin{prop}
    If $C_S\subset \mathbb{F}_{q^u}^n$ is the linear code generated by vectors in $S$, we have $C^\perp = \mathrm{Tr}(C_S)$, where $\mathrm{Tr}(C_S)$ denotes the trace code of $C_S$.
\end{prop}
\begin{proof}
    First we have $\mathbf{x}\in C^\perp_S$ for all $\mathbf{x}\in C$ according to the first property of $S$, which means $C\subset C^\perp_S\cap \mathbb{F}_q^n$; for all $\mathbf{z}\in C_S^\perp \cap \mathbb{F}_q^n$ again we use the first property and obtain $\mathbf{z}\in C$. Finally we get $C=C_S^\perp \cap \mathbb{F}_q^n=C_S^\perp|_ {\mathbb{F}_q^n}$. Now we use Delsarte's theorem and obtain $\left(C_S^\perp|_ {\mathbb{F}_q^n}\right)^\perp = \mathrm{Tr}\left((C_S^\perp)^\perp\right)$. This induces $C^\perp = \mathrm{Tr}(C_S)$.
\end{proof}
Using the concept of a test set, we propose the following general algorithm framework. The test set $S$ can be precomputed and is assumed to be available when a vector $\textbf{x}$ is received.

\begin{algorithm}[H]
\caption{DetermineCodeword}
\label{alg:deter}
\begin{algorithmic}[1]
\Require $\mathbf{x}\in \mathbb{F}_q^n$
\Ensure True or False 
\For{$i = 1$ to $R$}
    \State Randomly select a vector $\mathbf{y}$ from $S$
    \State Compute inner product $z = (\mathbf{x}, \mathbf{y})$
    \If{$z \neq 0$}
        \State \Return False
    \EndIf
\EndFor
\State \Return True
\end{algorithmic}
\end{algorithm}

\begin{prop}
    Algorithm \ref{alg:deter} can solve the decision problem “$\mathbf{x}\overset{?}{\in} C$" in $O(Run)$ operations, using $O(\left(\sharp S)un\right)$ space and with an error probability at most $p^R$.
\end{prop}
\begin{proof}
    Since the algorithm performs only one inner product operation per round, and noting that elements in $\mathbb{F}_{q^u}$ can be viewed as polynomials over $\mathbb{F}_q$, the computation of the inner product between vectors in $\mathbb{F}_{q^u}^n$ and $\mathbb{F}_q^n$ does not involve multiplication operations in $\mathbb{F}_{q^u}$. Therefore, each round of the algorithm executes $O(un)$ operations. With at most $R$ rounds of execution, the total number of operations is at most $O(Run)$. The set $S$ contains a total of $(\sharp S)n$ elements from $\mathbb{F}_{q^u}$. Therefore, storing $S$ requires $O((\sharp S) n  u)$ space in terms of $\mathbb{F}_q$ elements. Finally, since the proportion of vectors orthogonal to $\mathbf{x}$ is at most $p$, the failure probability of each round of the algorithm is at most $p$ when $\mathbf{x}\not \in C$. Thus, the error probability after $R$ independent repetitions is at most $p^R$. When $\mathbf{x}\in C$, the algorithm will not make mistakes.
\end{proof}
Now we present an explicit construction method for the test set $S$. This method essentially increases the proportion of non-orthogonal vectors by encoding the syndrome. Let $H\in \mathbb{F}_q^{(n-k)\times n}$ be the parity check matrix of $C$, $\bar{G}$ is the generator matrix of a $[\bar{n},n-k,\bar{d}]$ code over $\mathbb{F}_{q^u}^{\bar{n}}$. Let $S$ be the set of all row vectors of $\bar{G}^TH$, we have 
\[
\mathbf{x}\in C\Leftrightarrow H\mathbf{x}^T=\mathbf{0}\Leftrightarrow \bar{G}^TH\mathbf{x}^T=\mathbf{0}
\]
and
\begin{equation}
    \frac{\sharp\{\mathbf{y}\in S\;|\;(\mathbf{x},\mathbf{y})\neq 0\}}{\sharp S}=\frac{\mathrm{wt}(\bar{G}^TH\mathbf{x}^T)}{\bar{n}}\ge \frac{\bar{d}}{\bar{n}}.\label{2}
\end{equation}
Therefore $S$ is a test set with designed probability $1-\bar{d}/\bar{n}$. This method of constructing $S$ is efficient, as pre-computing $S$ involves just one matrix multiplication.  By carefully choosing the parameter $\bar{n}$  and $\bar{d}$, Algorithm \ref{alg:deter} can significantly improve the efficiency of the verification. Two different kinds of $S$ are obtained by choosing different kinds of $\bar{G}$ which will be illustrated below.


\paragraph{$\bar{G}$ comes from MDS codes}
For a fixed $u$ satisfying $\frac{n-k-1}{q^u-1}< 1$, arbitrarily choose $p$ within the range $\left(\frac{n-k-1}{q^u-1}, 1\right)$.
 Set $m=\left \lceil (n-k-1)/p \right \rceil $ and construct an MDS code with parameter $[m,n-k,m-(n-k)+1]$ over $\mathbb{F}_{q^u}$. Let $\bar{G}$ be its generator matrix, the designed probability of test set $S$ constructed by row vectors of $\bar{G}^TH$ is 
\[
1-\frac{m-(n-k)+1}{m}=\frac{n-k-1}{\left \lceil \frac{n-k-1}{p}  \right \rceil }\le p
\]
by inequality (\ref{1}) and (\ref{2}). In other words, given any $p$ in $\left(\frac{n-k-1}{q^u-1}, 1\right)$, we can construct a test set with designed probability $p$. This is the reason that we refer to $p$ as the designed probability.

\paragraph{$\bar{G}$ comes from asymptotically good codes}
Now we turn to solve the problem \ref{prob1.1}, which requires choosing $\bar{G}$ as generator matrix from asymptotically good codes.

\begin{theorem}\label{mainthm}
    Let $\{C_i\}$ be a family of asymptotically good $[n_i,k_i,d_i]$ codes over $\mathbb{F}_{q^u}$, which means \[\liminf_{i\to \infty} \frac{k_i}{n_i}\ge \alpha,\;\liminf_{i\to \infty}\frac{d_i}{n_i}\ge \beta,\;0<\alpha,\beta<1.\] We further assume there is a constant $W$ such that $n_i \le Wn_{i-1}$ for all $i$. By using generator matrices of $C_i$ to construct nice test sets, Algorithm \ref{alg:deter} can solve the decision problem  “$\mathbf{x}\overset{?}{\in} C$" in $O(un\log n)$ operations, using $O(un^2)$ space and with an error probability bounded by $1/\mathrm{poly}(n)$.
\end{theorem}
\begin{proof}
Since $C_i$ is asymptotically good, we may assume
\begin{equation}
\alpha-\epsilon \le \frac{k_i}{n_i},\;\beta-\epsilon\le \frac{d_i}{n_i}\label{3}    
\end{equation}
for all $i$, where $\epsilon$ is in $(0,\min\{\alpha,\beta\})$. Given any $[n,k]$ code $C$ with parity check matrix $H$, and assume $k_{i-1}<n-k\le k_i$ for some $i$, the generator matrix $\bar{G}_i$ of $C_i$ has shape $k_i\times n_i$. We delete some rows of $\bar{G}_i$ and convert it into a $(n-k)\times n_i$ matrix $\bar{G}_i'$ (when $n-k=k_i$, $\bar{G}_i'=\bar{G}_i$), and define the test set $S$ as the set of row vectors of $\bar{G}_i'^TH$. The designed probability $p$ of $S$ satisfies
\[
p\le 1-\frac{d_i}{n_i}\le 1-\beta +\epsilon <1
\]
by (\ref{1}) and (\ref{3}). Take the number of rounds $R=R_1\log n$, where $R_1$ is a constant and will be determined soon, the error probability after $R$ rounds is 
\[
p^R\le (1-\beta+\epsilon)^{R_1\log n} =\frac{1}{n^{R_1\log(1-\beta+\epsilon)^{-1}}},
\]
where $\log(1-\beta+\epsilon)^{-1}>0$. Therefore the error probability can be bounded by $1/\mathrm{poly}(n)$ by appropriately increasing the constant $R_1$ such that $R_1\log (1-\beta-\epsilon)^{-1}$ is larger than $\deg \mathrm{poly}(n)$. The total number of operations is obviously $O(unR_1\log n)=O(un\log n)$. Finally we show that the space complexity $O((\sharp S)un)=O(un^2)$ by
\[
\sharp S=n_i\le Wn_{i-1}\le W\frac{k_{i-1}}{\alpha-\epsilon}<\frac{W}{\alpha -\epsilon}(n-k).
\]
    
\end{proof}
Notably, $C_i$ can be constructed using the Garcia-Stichtenoth (GS) tower. More details on $C_i$ can be found in \cite{Garcia1995} and chapters 7 and 8 in \cite{Stichtenoth2009}. The GS tower requires the finite field size to be a square, so we choose $u=2$. Then the algorithm will satisfies all conditions in problem \ref{prob1.1}.

Our algorithm heavily relies on randomly selecting $O(\log n)$ vectors from $S$. This randomness guarantees an $O(n \log n)$ time complexity and a low error probability. For a general linear code $C$ and an arbitrary vector $\textbf{x}$, we clearly cannot determine whether $\textbf{x}$ is in $C$ by testing its orthogonality with only $O(\log n)$ vectors. Thus, derandomizing this procedure to obtain a deterministic algorithm presents an interesting and challenging problem.
\begin{prob}
    Is there a deterministic algorithm that solves the codeword decision problem for any linear code in $O(n\log n)$ time?
\end{prob}

\section{Applications}
In this section, we present some applications of our algorithm. 

\subsection{List decoding RM codes}
The first application is to list decoding of RM codes. We follow Wu's idea~\cite{1278668}, which views RM codes as subfield subcodes of RS codes.

\begin{lemma}[\cite{11197050}]
    Assume $r=u(q-1)+s$, where $0\le s<q-1$. Let $k=q^m-q^{m-u-1}(q-s)+1$. The Reed-Muller code $\mathrm{RM}_q(r,m)$ of length $N=q^m$ can be viewed as a subfield subcode of $\mathrm{RS}_{q^m}(q^m,k)$. Moreover, $\mathrm{RM}_q(r,m)$ has the same minimum distance $D=q^{m-u-1}(q-s)$ as $\mathrm{RS}_{q^m}(q^m,k)$.
\end{lemma}
To list decode $\mathrm{RS}_{q^m}(q^m,k)$ effectively, we also need the RS list-decoder proposed in \cite{Cohn2015} which achieves quasi-linear complexity.
\begin{lemma}[\cite{Cohn2015} \cite{11197050}]
    A $q$-ary $[n,k]$ Reed-Solomon code $\mathrm{RS}(n,k)$ is $(\rho,l)$-list decodable with decoding radius $\rho$ and list size $l$ given by
    \[
    \rho=1-\sqrt{1-\delta}-\epsilon,\;l=O(1/\epsilon),
    \]
    where $\epsilon>0$ is a small real number and $\delta = (n-k+1)/n$. The decoding complexity is $O(\epsilon^{-\omega-1}\log^2 n\log\log n)$, where $\omega$ is the exponent of matrix multiplication.
\end{lemma}
A slight modification to Algorithm 2 in \cite{11197050} leads to the following improved algorithm.
\begin{algorithm}[H]
\caption{RM-list-Decoder($\rho,l$)}
\begin{algorithmic}[1]
\Require $\mathbf{r}\in \mathbb{F}_q^N$
\Ensure $L=\{\mathbf{a}\in \mathbb{F}_q^N\;|\; \mathrm{wt}(\mathbf{a}-\mathbf{r})\le \rho N\}$
\State Put $L=\emptyset$
\State Call RS-list-decoder$(\rho,l)$ with input $\mathbf{r}$. Let $L'$ be its output.
\For{ $\mathbf{c}\in L'\cap \mathbb{F}_q^N$} 
    \If{DetermineCodeword($\mathbf{c}$)}
        \State $L=L\cup \{\mathbf{c}\}$
    \EndIf
\EndFor
\State \Return $L$
\end{algorithmic}
\end{algorithm}
Similarly, Step 2 will cost $O(\mathrm{M}_q(m)\cdot \frac{1}{\epsilon^{\omega+1}}N\log^2N\log\log N )$ operations, where $\mathrm{M}_q(m)$ stands for the number of operations in $\mathbb{F}_q$ to multiply two univariate polynomials over $\mathbb{F}_q$ of degree less than $m$. Step 4 requires $O(N\log N)$ operations according to Theorem \ref{mainthm} with an error probability bounded by $1/\mathrm{poly}(N)$. Therefore, the probability that the output list $L$ contains non-Reed-Muller codewords is less than $1-(1-1/\mathrm{poly(N)})^l\sim  l/\mathrm{poly}(n)$.

\subsection{Comparison with specialized algorithm for RS codes}
In this subsection we compare our algorithm with other deterministic syndrome computation algorithms. These algorithms are specifically designed for RS codes and are primarily based on FFT techniques \cite{4395269} \cite{6176255} \cite{5910103} \cite{9925209}. For RS($255,223$) and RS($1023,895$), we encode the syndrome with RS($255,32$) and RS($1023,128$), hence the designed probability is $31/255$ and $127/1023$ respectively. By taking $7$ independent rounds the error probability is less than $4\times10^{-7}$ and $4.6\times 10^{-7}$. It can be observed that our algorithm matches that of the specialized algorithms in the problem of RS codeword verification under these settings. 

\begin{table}[h]
\centering
\caption{Complexity comparison with specialized syndrome computation algorithms}
\begin{tabular}{|l|ll|ll|}
\hline
\multirow{2}{*}{} & \multicolumn{2}{l|}{RS($255,223$)} & \multicolumn{2}{l|}{RS($1023,895$)} \\ \cline{2-5} 
                  & \multicolumn{1}{l|}{Mul.}  & Add.  & \multicolumn{1}{l|}{Mul.}   & Add.  \\ \hline
Method in  \cite{4395269}       & \multicolumn{1}{l|}{3,060}  & 4,998  & \multicolumn{1}{l|}{33,620}  & 73,185 \\ \hline
Method in  \cite{6176255}      & \multicolumn{1}{l|}{252}   & 3,064  & \multicolumn{1}{l|}{2,868}   & 19,339 \\ \hline
Method in    \cite{5910103}     & \multicolumn{1}{l|}{149}   & 2,931  & \multicolumn{1}{l|}{824}    & 36,981 \\ \hline
Method in    \cite{9925209}     & \multicolumn{1}{l|}{752}   & 1,696  & \multicolumn{1}{l|}{4,160}   & 9,088  \\ \hline
Our method        & \multicolumn{1}{l|}{1,785}      &   1,778    & \multicolumn{1}{l|}{7,161}       &  7,154     \\ \hline
\end{tabular}

\end{table}

\bibliographystyle{alpha}
\bibliography{sample}

\end{document}